\documentclass{style}
\pdfoutput=1
\usepackage{cite}

\begin{document}
\title{Approximating Smallest Containers for Packing Three-dimensional Convex Objects}

\author{Helmut Alt%
  \thanks{\texttt{alt@mi.fu-berlin.de}} }

\author{Nadja Scharf%
  \thanks{\texttt{nadja.scharf@fu-berlin.de} \\This research was partially funded by DFG (Deutsche Forschungsgemeinschaft) under grant no.~AL253/7-2.}}
\affil{Institute of Computer Science, Freie Universität Berlin, Takustr.~9, 14195~Berlin, Germany}
	
	\maketitle
	
	\begin{abstract}
	We investigate the problem of computing a minimal-volume container for the non-overlapping packing of a given set of three-dimensional convex objects. Already the simplest versions of the problem are $\mathcal{NP}$-hard so that we cannot expect to find exact polynomial time algorithms. We give constant ratio approximation algorithms for packing axis-parallel (rectangular) cuboids under translation into an axis-parallel (rectangular) cuboid as container, for cuboids under rigid motions into an axis-parallel cuboid or into an arbitrary convex container, and for packing convex polyhedra under rigid motions into an axis-parallel cuboid or arbitrary convex container. This work gives the first approximability results for the computation of minimal volume containers for the objects described.

	\end{abstract}
	\section{Introduction}	The problem of efficiently packing objects arises in a large variety of contexts. Apart from the obvious ones, where concrete objects need to be packed for transportation or storage, there are more abstract ones, for example cutting stock or scheduling. Given a set of objects that have to be cut out from the same material the objective is to minimize the waste, i.e., place the pieces to be cut out as close as possible. In the case of scheduling, a list of jobs is given. Each job needs a certain amount of given resources and the aim is to minimize under certain constraints this need of resources such as time, space,  or number of machines. Altogether, this situation can be described as a problem of packing high-dimensional cuboids into a strip with bounded side lengths.
So,  both problems can be viewed as a given list of objects for which a container of minimal size is wanted.

In this work, we consider the more general and abstract problem of packing three-dimensional convex polyhedra into a minimum volume container. All variants of this problem are NP-hard and we will develop constant factor approximation algorithms for some of them. The worst case constant factors are still very high, but probably they will be much lower for realistic inputs. The major aim of this paper, however, is to show the existence of constant factors at all, i.e., that the problems belong to the complexity class APX.

	\subsubsection*{Related Work}
	So far, there are only few results about finding containers of minimum volume. Related problems include strip packing and bin packing. In two-dimensional strip packing the width of a strip is given and the objects should be packed in order to minimize the length of the strip used. In three dimensions, the rectangular cross section of the strip is fixed. Bin-packing is the problem where the complete container is fixed and the objective is to minimize the number of containers to pack all objects. For both problems usually only translations are allowed to pack the objects. 
	
For two-dimensional bin packing there exists an algorithm with an asymptotic approximation ratio of 1.405~\cite{2D-BP_asym1.405} and Bansal et al.\ proved that there cannot be an APTAS unless $\mathcal{P=NP}$~\cite{2D-BP_NoAPTAS}. For two-dimensional strip packing there exists an AFPTAS~\cite{2D-SP_AFPTAS}. In three dimensions there are algorithms with an asymptotic approximation ratio of 4.89 for bin packing~\cite{3D-BP_4.89} and an asymptotic approximation ratio of $\frac{3}{2} + \varepsilon$ for strip packing~\cite{JansenPraedel}. The best known worst case approximation ratio for three-dimensional strip packing is $\frac{29}{4}$~\cite{Knapsack}.
	
	For two dimensions, von Niederhäusern~\cite{MASTER} gave algorithms for packing rectangles or convex polygons in a minimal-area rectangular container with approximation ratios 3 and 5 respectively. A recent result shows that packing convex polygons under translation into a minimum-area rectangular or convex container can be approximated with ratios $17.45$ and $27$ respectively~\cite{ConvexPolygons}.
	
	\textsc{PARTITION} can be reduced to one-dimensional bin packing and one\hyp dimensional bin packing is a special case of higher dimensional bin or strip packing. If one-dimensional bin packing could be approximated with a ratio smaller than $\frac{3}{2}$, we could solve \textsc{PARTITION}. Therefore, none of the mentioned problems can be approximated better than with ratio $\frac{3}{2}$ unless $\mathcal{P}= \mathcal{NP}$. \textsc{PARTITION} can also be reduced to our problem showing $\mathcal{NP}$-hardness.
	\subsubsection*{Our Results}
	In this work we give the first approximation results for packing three-dimensional convex objects in a minimum-volume container. For packing axis-parallel rectangular cuboids under translation into an axis-parallel rectangular cuboid as a container, we achieve a $7.25+\varepsilon$ approximation. If we allow the cuboids to be packed under rigid motions (translation and rotation) then we achieve an approximation ratio of $17.737$ for an axis-parallel cuboid as container and an approximation ratio of $29.135$ for an arbitrary convex container. For packing convex polyhedra under rigid motions we achieve an approximation ratio of $277.59$ for computing an axis-parallel cuboid as container and $511.37$ for a convex container.
	
	\section{Preliminaries and Notation}
	For most algorithms considered here, the input is a set of rectangular boxes $\mathcal{B} = \left\lbrace b_1, b_2, \dots b_n\right\rbrace$. We denote a box $b_i$ in axis-parallel orientation by a tuple of its height, width and depth $\left(h_i, w_i, d_i\right)$. We denote by $h_{\max} = \max \left\lbrace h_i \mid b_i \in \mathcal{B}\right\rbrace$, $w_{\max} = \max \left\lbrace w_i \mid b_i \in \mathcal{B}\right\rbrace$ and $d_{\max} = \max \left\lbrace d_i \mid b_i \in \mathcal{B}\right\rbrace$.

For points $P$ and $Q$ we denote by $\overline{PQ}$ the line segment between $P$ and $Q$ of length $\vert PQ\vert$. $\overrightarrow{PQ}$ denotes the vector from $P$ to $Q$.
When we write "axis-parallel container" we mean "axis-parallel rectangular cuboid as a container". We use the term box as a synonym for rectangular cuboid.
\begin{newdef}[strip packing]
An instance for the \emph{strip packing problem} consists of an axis parallel strip with all dimensions fixed except for one, and a set of axis parallel boxes. Call the open dimension the height. The aim is to pack the boxes under translation into the strip such that the used height gets minimized. The boxes are not allowed to overlap.
\end{newdef} 
\begin{newdef}[orthogonal minimal container packing---\textsc{OMCOP}] 
An instance of this problem is a set of convex polyhedra. The aim is to pack these polyhedras non-overlapping such that the minimal axis-parallel container has minimal volume. Variants include the kind of motions allowed or that more specialized objects are to be packed.
\end{newdef}
This work only considers algorithms in two or three dimensions. For ease of notation we always assume the lower left (front) corner of the container to lie in the origin. $V_{\text{opt}}$ denotes the minimal possible volume for a container.

The following algorithm was given by von Niederhäusern~\cite{MASTER}. It will be used later as a subroutine. For an example see Figure~\ref{fig:Algo2D}.

\begin{algorithm}[H]
\KwIn{A list $\mathcal{S}$ of rectangles $r_i$, denoted by their width $w_i$ and height $h_i$, a width for the strip $w$}
\nl Order the rectangles in $\mathcal{S}$ by decreasing width, such that if $i<j$ then $w_i  \geq w_j$.\\
\nl Split $\mathcal{S}$ in sublists $\mathcal{S}_j = \left\lbrace r_i \in \mathcal{S} \mid \frac{w}{2^{j-1}} \geq w_i > \frac{w}{2^{j}} \right\rbrace$ for $j\geq 1$.\\
\nl Start with packing the rectangles in $\mathcal{S}_1$ on top of each other in the strip $\left[0, w\right] \times \left[0, \infty\right)$.\\
\nl Split the remaining strip in two substrips with width $\frac{w}{2}$ and pack the rectangles in $\mathcal{S}_2$ one after another into these substrips. $r_i$ is packed in the substrip with current minimal height.\\
\nl Again split the substrips into two and pack $S_3$. Iterate that process until everything is packed.
\caption{}
\label{alg:3Appr2D}
\end{algorithm}

\begin{remark}
Note that the strip is half filled with rectangles up to the lower boundary of the highest rectangle that touches the upper end of the packing. Otherwise,this rectangle could have been placed lower. That means that the strip is half filled with rectangles except for a part with area at most $w\cdot h_{\max}$.
\label{rem:halfFilled}
\end{remark}
\begin{remark}
Steps 1 and 2 can be done in $\mathcal{O}\left(n \log n\right)$ time where $n$ is the size of $\mathcal{S}$. If we store all substrips in a height-balanced tree and split a strip only in two if it gets used, we can perform steps 3 to 5 in $\mathcal{O}\left(n \log n\right)$.
\label{rem:runtime2Dpack}
\end{remark}
\begin{figure}[ht]
\centering
\includegraphics[scale=0.4]{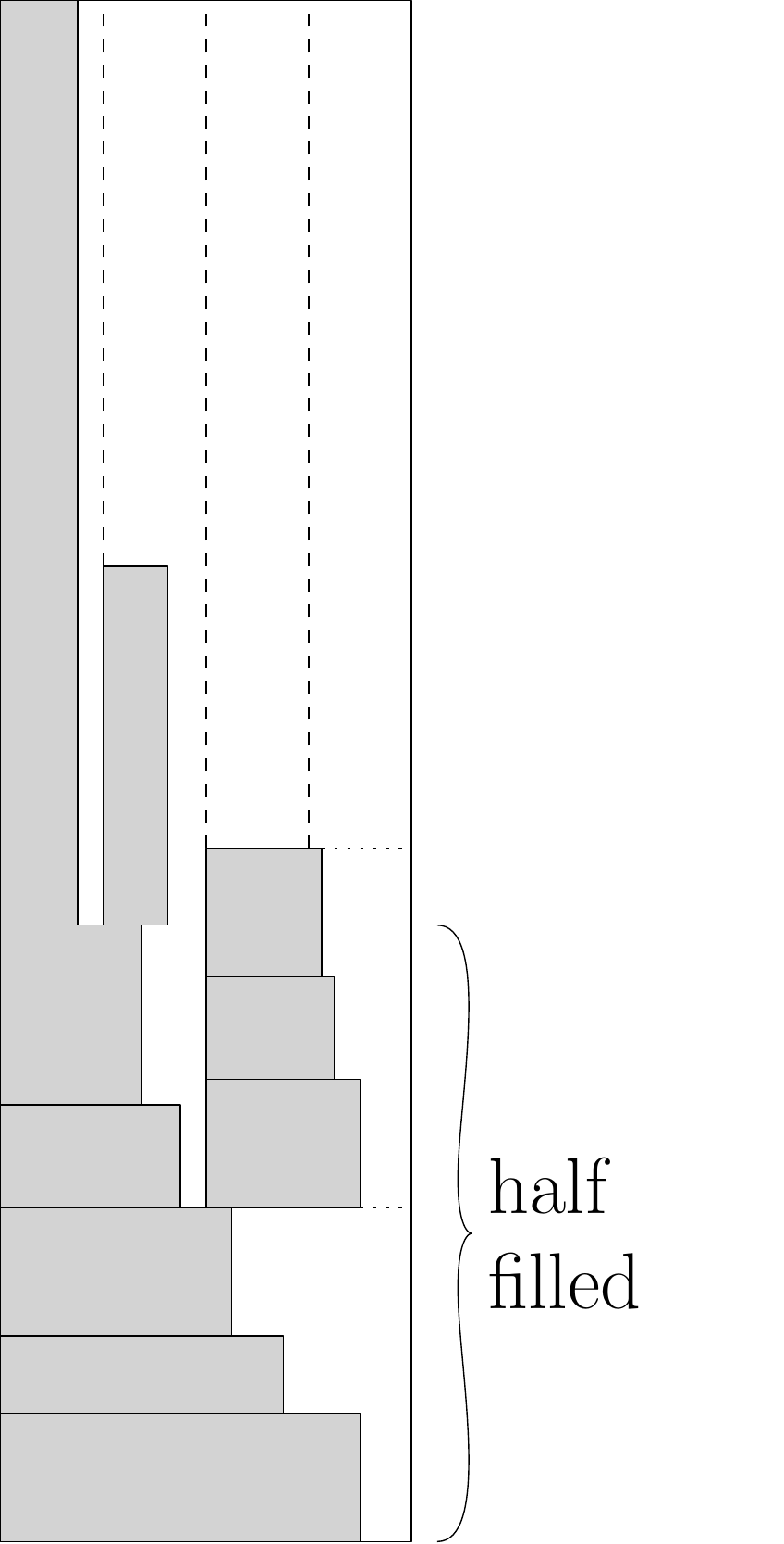}

\caption{Result of Algorithm~\ref{alg:3Appr2D}}
\label{fig:Algo2D}
\end{figure}

	\section{Reduction from 3D-\textsc{OMCOP} to Strip Packing}
	In this section we consider the version of \textsc{OMCOP} where the given objects are axis-parallel boxes that are to be packed under translation. The idea behind the reduction of \textsc{OMCOP} to strip packing is to test different base areas for the strip and to return the result with minimal volume. Assuming that the lower left corner of the base area is located at the origin, we test each point in a set $\mathcal{S}$ as a possible upper right corner for the base area. Testing means that we call a strip packing algorithm with the given boxes and the base area implied by the point of $\mathcal{S}$. $\mathcal{S}$ will be determined by a parameter $\varepsilon$: the smaller $\varepsilon$, the more elements $\mathcal{S}$ contains, the better the approximation ratio gets.

Note that for the width $W_{\text{opt}}$ of an optimal container, the following inequalities hold: 
\begin{enumerate}
\item $W_{\text{opt}} \leq W_{\sum}$, where $W_{\sum}$ denotes the sum of all widths of the boxes to be packed. It is an upper bound because the width of an optimal container has to be the sum of width of some of the objects. Otherwise they can be pushed together reducing the width of the container and thereby its volume.
\item $W_{\text{opt}} \geq w_{\max}$, where $w_{\max}$ denotes the width of the widest box. Since this box needs to be packed, this is a lower bound for the width of the container.
\end{enumerate}
The analogous bounds for the depth of an optimal container hold for the same reasons. In the following $H_{\text{opt}}$, $W_{\text{opt}}$ and $D_{\text{opt}}$ denote the height, width and depth of the same optimal container. Let $\varepsilon'= \frac{\varepsilon}{2\left(\varepsilon + \alpha\right)}$ for a constant $\alpha$ defined later.

The set $\mathcal{S}$ is obtained by dividing the intervals of possible width and depth logarithmically.
\begin{align*}
\mathcal{S} =
&\lbrace W_{\sum}\left(1-\varepsilon'\right)^i 
\mid i \in \mathbb{N}, 
W_{\sum}\left(1-\varepsilon'\right)^i > w_{\max} \rbrace 
\cup \lbrace w_{\max}\rbrace 
\times\\ 
&\lbrace D_{\sum}\left(1-\varepsilon'\right)^j
\mid j \in \mathbb{N}, 
D_{\sum}\left(1-\varepsilon'\right)^j > d_{\max} \rbrace 
\cup \lbrace d_{\max}\rbrace.
\end{align*}
For an example for $\mathcal{S}$ see Figure~\ref{fig:exampleS}.
\begin{figure}
\centering
\includegraphics[height=0.37\textwidth]{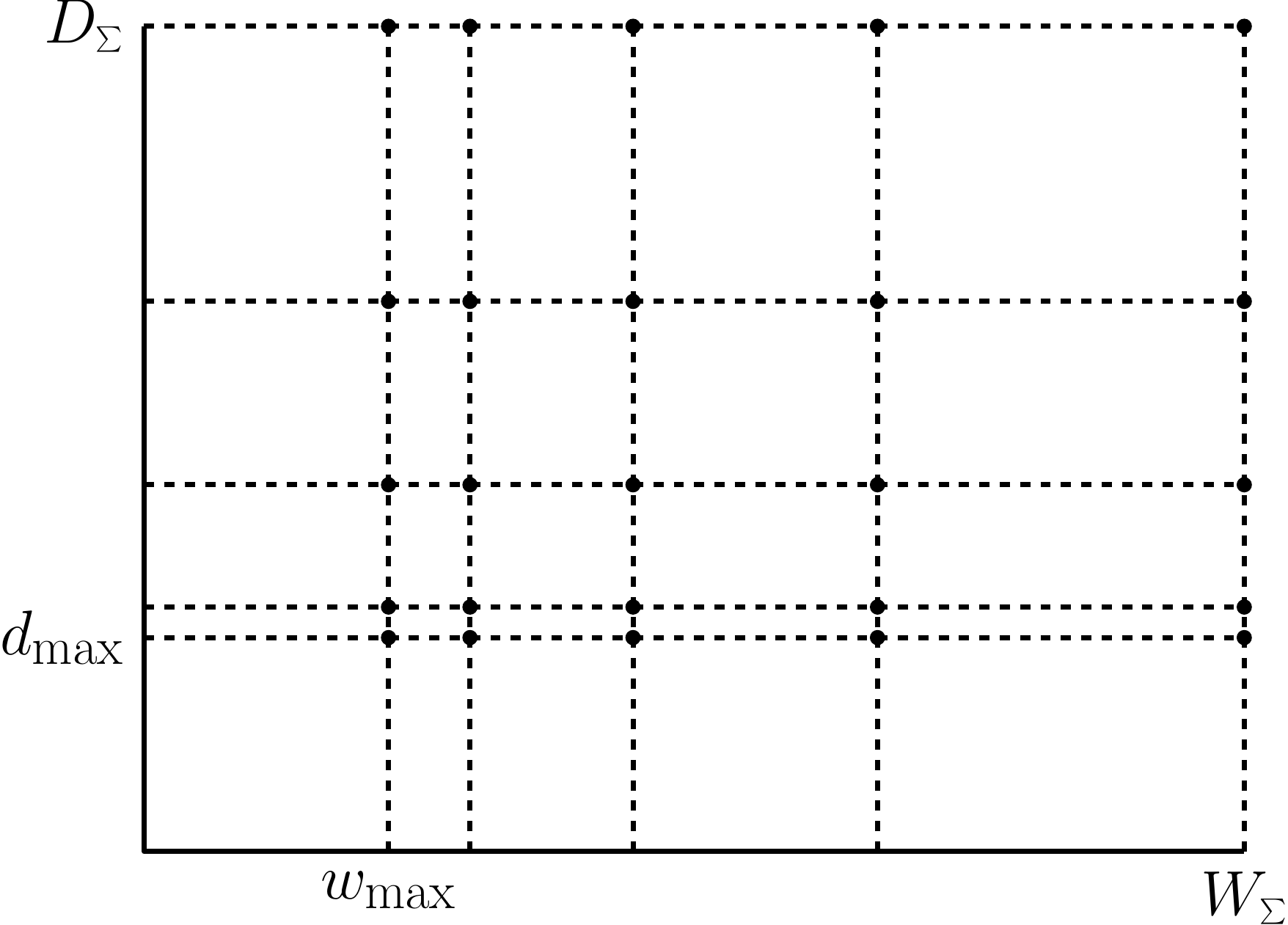}
\caption{Example for Set $\mathcal{S}$ with $\varepsilon = \frac{3}{4}$ and $\alpha=1.5$}

\label{fig:exampleS}
\end{figure}
\begin{thm}
If we use an $\alpha$-approximation algorithm to pack the boxes under translation into the strips and the set $\mathcal{S}$ defined above, we obtain an $\left(\alpha + \varepsilon\right)$-approximation for the \textsc{OMCOP} variant where $n$ axis aligned boxes are to be packed under translation. Its runtime is $\mathcal{O}\left(T(n)\frac{\log^2n}{\varepsilon^2}\right)$ where $T(n)$ is the runtime of the $\alpha$-approximation algorithm for strip packing.
\end{thm}
\begin{proof}
There exist $a,b \in \mathbb{N}$ with 
$W_{\sum}\left(1-\varepsilon'\right)^{a+1} < W_{\text{opt}} \leq W_{\sum}\left(1-\varepsilon'\right)^{a}$ 
and $D_{\sum}\left(1-\varepsilon'\right)^{b+1} < D_{\text{opt}} \leq D_{\sum}\left(1-\varepsilon'\right)^{b}$. 
Eventually the boxes will be packed in a strip with base area $W \times D$ with $W=W_{\sum}\left(1-\varepsilon'\right)^{a}$ and $D=W_{\sum}\left(1-\varepsilon'\right)^{b}$. Since the extensions of the base area are at least the ones for an optimal container, we obtain a packing with height $H \leq \alpha H_{\text{opt}}$. The associated container has volume $V$ with
\begin{align*}
V &= HWD\\
&\leq \left(\alpha H_{\text{opt}}\right)\left(W_{\sum}\left(1-\varepsilon'\right)^{a}\right)\left(D_{\sum}\left(1-\varepsilon'\right)^{b}\right)\\
&\leq \left(\alpha H_{\text{opt}}\right)\left(\frac{W_{\text{opt}}}{1-\varepsilon'}\right)\left(\frac{D_{\text{opt}}}{1-\varepsilon'}\right)\\
&\leq \frac{\alpha}{\left(1-\varepsilon'\right)^2} V_{\text{opt}}\\
&\leq \frac{\alpha}{1-2\varepsilon'} = \left(\alpha + \varepsilon\right) V_{\text{opt}} &\text{, since } \varepsilon'= \frac{\varepsilon}{2\left(\varepsilon + \alpha\right)}
\end{align*}
The size of $\mathcal{S}$ is 
\begin{align*}
\vert\mathcal{S}\vert &= \left(\left\lceil \log_{\frac{1}{1-\varepsilon'}} W_{\sum} \right\rceil - \left\lfloor \log_{\frac{1}{1-\varepsilon'}} w_{\max} \right\rfloor +1 \right)\left(\left\lceil \log_{\frac{1}{1-\varepsilon'}} D_{\sum} \right\rceil - \left\lfloor \log_{\frac{1}{1-\varepsilon'}} d_{\max} \right\rfloor +1 \right)\\
&= \mathcal{O}\left(\frac{\log^2 n}{\left(-\log\left(1-\varepsilon'\right)\right)^2}\right) \text{, since } \frac{W_{\sum}}{w_{\max}}\leq n \text{, where $n$ is the number of boxes}\\
&= \mathcal{O}\left(\frac{\log^2 n}{\varepsilon^2}\right) \text{, since } -\log\left(1-x\right) \geq x \text{ for } x \in \left[0,1\right] \text{ and } \varepsilon' \in \Theta\left(\varepsilon\right),
\end{align*} 
and therefore we get the desired running time.
\end{proof}

If we use the algorithm given by Diedrich et al.~\cite{Knapsack} to pack the boxes into the strips, we obtain the following corollary.
\begin{cor}
There exists a $\left(7.25 + \varepsilon\right)$-approximation algorithm for packing axis-parallel boxes under translation into a minimal volume axis-parallel box with running time polynomial in both the input size and $\frac{1}{\varepsilon}$.
\end{cor}
	
	\section{Algorithms for Variants of \textsc{OMCOP}}
	In this section, we will give algorithms for variants of \textsc{OMCOP}. The basic idea is to get rid of the third dimension by dividing the set of objects into sets of objects with similar height and then packing those using an algorithm for two-dimensional boxes. These containers then get cut into pieces with equal base area and the pieces will be stacked on top of each other.
\subsection{Packing Cuboids under Translation}
Even though this algorithm gets outperformed by the construction in the previous section, we state it here as base for the algorithms for the other variants.
\begin{algorithm}[H]
\KwIn{Set of axis parallel boxes $\mathcal{B}=\left\lbrace b_1, \dots, b_n\right\rbrace$, parameter $\varepsilon$, parameter $c$}
\nl Partition $\mathcal{B}$ into subsets of boxes that have almost the same height:\\
		$\mathcal{B}_j=\left\lbrace b_i \in \mathcal{B} \mid h_{\max}\left(1-\varepsilon\right)^j < h_i \leq h_{\max}\left(1-\varepsilon\right)^{j -1}\right\rbrace$.\\
\nl Use Algorithm~\ref{alg:3Appr2D} to pack the boxes of every $\mathcal{B}_j$ into a strip with width $w_{\max}$ and height $h_{\max}\left(1-\varepsilon\right)^{j -1}$ considering the base areas of the boxes.\label{step:pack}\\
\nl Divide the strips into pieces with depth $\left(c-1\right) \cdot d_{\max}$, ignoring the last part of the strip of depth $d_{\max}$. (Parts of boxes contained in this part of the strip will be covered in step~\ref{step:extend} anyway.)\label{step:divide}\\
\nl Extend each piece to depth $c \cdot d_{\max}$ such that every box lies entirely in the piece its front lies in.\label{step:extend}\\
\nl Stack the pieces on top of each other\label{step:stack}.
\caption{}
\label{alg:withoutRotations}
\end{algorithm}

For an illustration of steps~\ref{step:divide}~to~\ref{step:stack} see Figure~\ref{fig:Algo}.\\

The first step can be done by sorting, so it needs $\mathcal{O}(n)$ time. The second step needs time $\mathcal{O}(n\log n)$ (see Remark~\ref{rem:runtime2Dpack}). The rest can be done in linear time. Therefore Algorithm~\ref{alg:withoutRotations} runs in $\mathcal{O}(n\log n)$ time.
\begin{figure}[ht]
    \centering
    \begin{subfigure}[b]{0.55\textwidth}
        \includegraphics[scale=0.45]{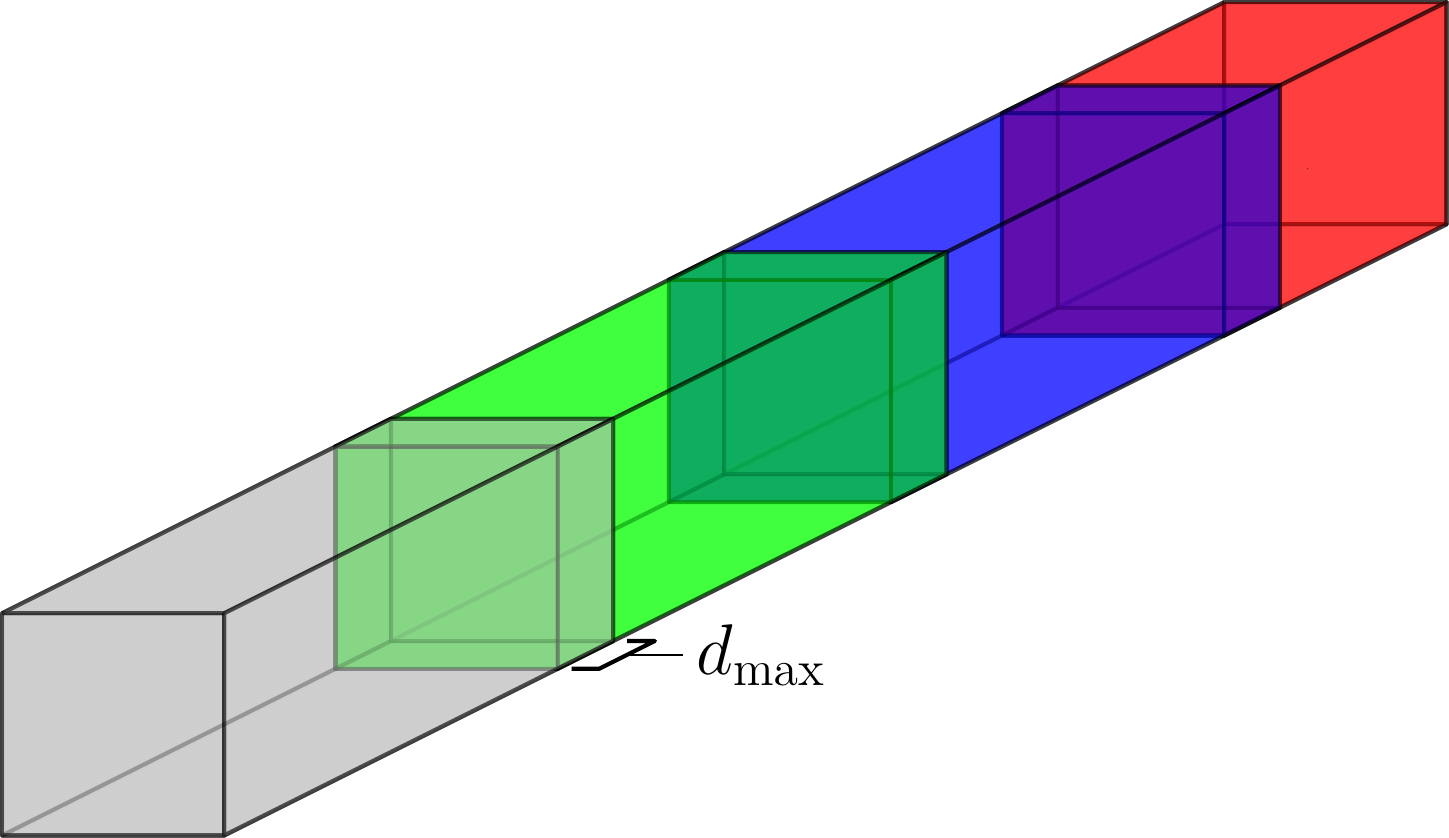}
        \caption{Cut strip}
    \end{subfigure}
    \begin{subfigure}[b]{0.4\textwidth}
        \includegraphics[scale=0.45]{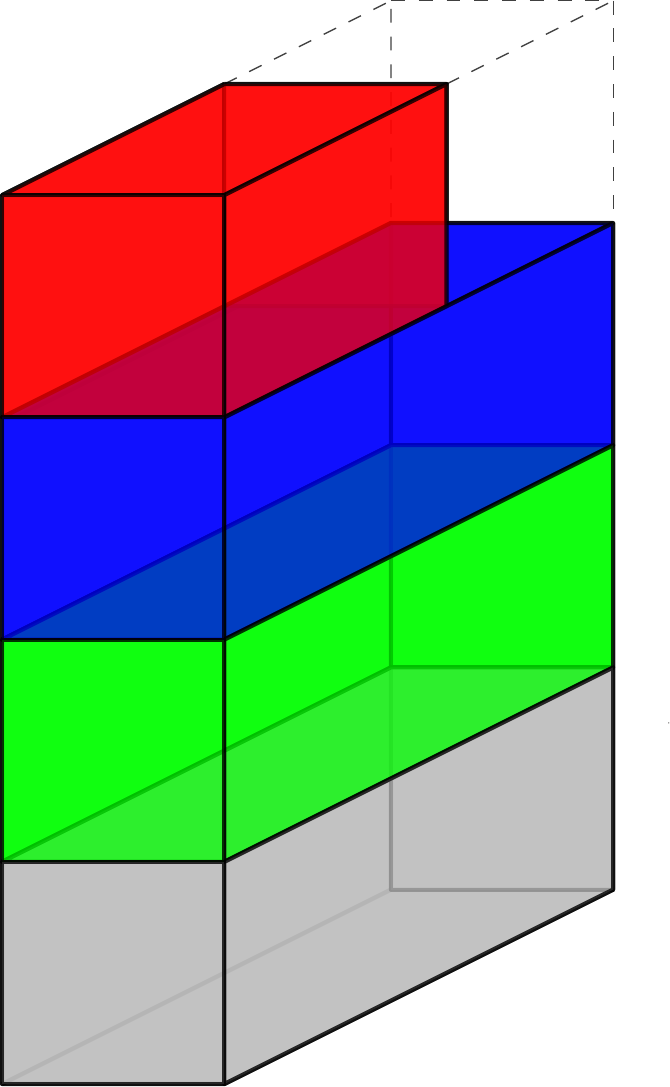}
        \caption{Pieces obtained from one strip stacked on top of each other}
    \end{subfigure}
  
    \caption{}
    \label{fig:Algo}
\end{figure}
\begin{thm}
For suitable values of $c$ and $\varepsilon$ Algorithm~\ref{alg:withoutRotations} computes a $\left(\frac{3}{\sqrt[3]{2}-1}\!\approx\! 11.542\right)$-approximation for the variant of three-dimensional \textsc{OMCOP} where $n$ axis parallel cuboids are packed under translation in $\mathcal{O}(n\log n)$ time.
\end{thm}
\begin{proof}
Let $D_j$ denote the depth of the strip obtained in step~\ref{step:pack} for the boxes in $\mathcal{B}_j$. Then we get by step~\ref{step:divide} $\left\lceil \frac{D_j - d_{\max}}{\left(c-1\right)d_{\max}}\right\rceil$ pieces. After step~\ref{step:extend} each piece has volume $c\cdot d_{\max} w_{\max} h_{\max}\left(1-\varepsilon\right)^{j -1}$. Consider the total volume $V_j$ of the pieces obtained for the subset $\mathcal{B}_j$:
\begin{align*}
V_j &= c \cdot d_{\max} \left\lceil \frac{D_j - d_{\max}}{\left(c-1\right)d_{\max}}\right\rceil w_{\max} h_{\max}\left(1-\varepsilon\right)^{j -1}\\
&< \frac{c}{c-1} \left(D_j-d_{\max}\right) w_{\max} h_{\max}\left(1-\varepsilon\right)^{j -1} + c \cdot d_{\max}w_{\max}h_{\max}\left(1-\varepsilon\right)^{j -1}.
\end{align*}
We know from the two-dimensional packing algorithm that the base area of the strip is half filled with boxes except for the last part of depth $d_{\max}$ (Remark~\ref{rem:halfFilled}), 
so $\left(D_j - d_{\max}\right)w_{\max} \leq 2 \sum_{b_i \in \mathcal{B}_j} A_B\left(b_i\right)$ 
where $A_B\left(b\right)$ denotes the base area of box $b$. We also know that for every $b_i \in \mathcal{B}_j$ the inequality $h_{\max}\left(1-\varepsilon\right)^{j -1} < \frac{h_i}{1-\varepsilon}$ holds. Therefore, we get for the total volume of the packing $V$ that
\begin{align}
V &\leq \sum_{j=1}^{\infty}\left(\frac{c}{c-1} \left(D_j-d_{\max}\right) w_{\max} h_{\max}\left(1-\varepsilon\right)^{j -1} + c \cdot d_{\max}w_{\max}h_{\max}\left(1-\varepsilon\right)^{j -1}\right)\nonumber\\
&\leq \sum_{j=1}^{\infty} \left(\frac{2c}{\left(1-\varepsilon\right)\left(c-1\right)} \sum_{b_i \in \mathcal{B}_j}V\left(b_i\right) + c \cdot w_{\max} \cdot d_{\max} \cdot h_{\max}\left(1-\varepsilon\right)^{j-1}\right)\nonumber\\
&\leq \frac{2c}{\left(1-\varepsilon\right)\left(c-1\right)} \underbrace{\sum_{b \in \mathcal{B}} V\left(b\right)}_{\leq V_{\text{opt}}} + 
c \cdot \underbrace{w_{\max} \cdot d_{\max} \cdot h_{\max}}_{\leq V_{\text{opt}}} \cdot \sum_{l=0}^{\infty}\left(1-\varepsilon\right)^l \label{eq:getAssympAppr}\\
&\leq \left(\frac{2c}{\left(1-\varepsilon\right)\left(c-1\right)}+ \frac{c}{\varepsilon}\right) V_{\text{opt}}\label{eq:ApproxPackAxisParBoxes}.
\end{align}
The factor before $V_{\text{opt}}$ in term~(\ref{eq:ApproxPackAxisParBoxes}) is minimized if the partial derivatives with respect to $c$ and $\varepsilon$ are 0. Solving the resulting system of equations we get $c = \sqrt[3]{2}+1$ and $\varepsilon = \frac{1}{3}\left(\sqrt[3]{4}-\sqrt[3]{2}+1\right)$. This gives an approximation ratio of $\frac{3}{\sqrt[3]{2}-1}$. 
\end{proof}

\subsection{Packing Cuboids under Rigid Motions}
\label{sec:PackCuboidRigidMotions}
Now we consider the variant of \textsc{OMCOP} where the objects to be packed are boxes and rigid motions are allowed. We basically use the algorithm stated above but with an extra preprocessing step, namely rotating every box $b_i \in \mathcal{B}$ such that it becomes axis parallel and $h_i \geq w_i \geq d_i$. This can be done in $\mathcal{O}(n)$ time. To prove the performance bound of this algorithm we need the following lemma.
\begin{lem}If every $b_i \in \mathcal{B}$ is oriented such that $h_i \geq w_i\geq d_i$, then \begin{equation*}
h_{\max}\cdot w_{\max} \cdot d_{\max} \leq \sqrt{6}\cdot V_{\text{opt}}.
\end{equation*}
\label{lem:AxisParContBoxes}
\end{lem}
\begin{proof}
Since an optimal container has to contain the box determining $h_{\max}$, it contains a line segment of length $h_{\max}$. The projection of that line segment on at least one of the axes has to have length at least $\frac{1}{\sqrt{3}}h_{\max}$. W.l.o.g. let this axis be the x-axis. Therefore, the optimal container has an expansion of at least $\frac{1}{\sqrt{3}}h_{\max}$ in x-direction.

Since every box is higher then wide, a box with width $w_{\max}$ contains a disk $D$ with diameter $w_{\max}$ and so the optimal container does. Observe that $D$ contains a diametric line segment $l$ which is parallel to the y-z-plane. Consequently, the projection of $l$ and therefore the one of the whole box on the y-axis or on the z-axis has a length of at least $\frac{1}{\sqrt{2}}w_{\max}$. W.l.o.g. let this be the y-axis.

A box with depth $d_{\max}$ contains a sphere with diameter $d_{\max}$. The projection of this sphere on any axis has length at least $d_{\max}$.

Summarizing, each optimal box has volume at least $\frac{1}{\sqrt{6}}h_{\max}\cdot w_{\max} \cdot d_{\max}$
\end{proof}

Observe that every argument leading to inequality~(\ref{eq:getAssympAppr}) still holds for this variant of the algorithm. Using Lemma~\ref{lem:AxisParContBoxes} to estimate $h_{\max}\cdot w_{\max} \cdot d_{\max}$ we get an approximation factor of
\begin{equation*}
\frac{2c}{\left(1-\varepsilon\right)\left(c-1\right)} + \frac{c \cdot \sqrt{6}}{\varepsilon}.
\end{equation*}
Minimizing this expression as before yields the following theorem.
\begin{thm}
The given algorithm computes a $17.738$-approximation for the variant of three-dimensional \textsc{OMCOP} where $n$ axis parallel cuboids are packed under rigid motions in $\mathcal{O}(n\log n)$ time.
\end{thm}

\subsubsection*{Convex Container}
If we allow a convex container instead of an orthogonal container, we can use the same algorithm but adapt the analysis. The arguments leading to inequality~(\ref{eq:getAssympAppr}) still hold since they only use the total volume of the boxes as estimate for the volume of an optimal container. To estimate $h_{\max}\cdot w_{\max} \cdot d_{\max}$, we use the following lemma. Note that $V_{\text{opt}}$ here denotes the volume of a minimal convex container instead of an axis parallel container.
\begin{lem}
If every $b_i \in \mathcal{B}$ is oriented such that $h_i \geq w_i\geq d_i$, then \begin{equation*}
h_{\max}\cdot w_{\max} \cdot d_{\max} \leq 6 \cdot V_{\text{opt}}.
\end{equation*}
\end{lem}
\begin{proof}
Consider the line segment, disk and sphere from the proof of Lemma~\ref{lem:AxisParContBoxes}. The line segment has length $h_{\max}$. The disk with diameter $w_{\max}$ contains a line segment of length $w_{\max}$ that is perpendicular to the first line segment. The sphere with diameter $d_{\max}$ contains a line segment of length $d_{\max}$ that is perpendicular to the first two line segments. It is well known (see, e.g., Lemma 6 from~\cite{convexRegions}) that the convex hull of these three line segments has a volume of at least $\frac{1}{6}h_{\max} w_{\max} d_{\max}$.
\end{proof}
This leads with inequality~(\ref{eq:getAssympAppr}) to the following approximation ratio:
\begin{equation*}
\frac{2c}{\left(1-\varepsilon\right)\left(c-1\right)} + \frac{c \cdot 6}{\varepsilon}.
\end{equation*}
Minimizing this term as before yields the following theorem.
\begin{thm}
Using the algorithm described in section~\ref{sec:PackCuboidRigidMotions} we get a $29.135$\hyp approximation for packing $n$ axis parallel boxes under rigid motions into a smallest-volume convex container in time $\mathcal{O}(n \log n)$.
\end{thm}
	
\subsection{Packing Convex Polyhedra under Rigid Motions}
\label{sec:ConvPolyOMCOP}
We use the algorithm from the previous sections to pack convex polyhedra under rigid motions into an axis-parallel box of minimal volume. To do so, we add another preprocessing step where we compute a bounding box for every polyhedron according to the following lemma. We then pack these boxes with the algorithm discussed in the previous section.
\begin{lem}
There is a box $B$ for every convex polyhedron $K$ in $\mathbf{R}^n$ that contains $K$ such that
\begin{equation*}
V(B) \leq n! V(K).
\end{equation*}
$B$ can be computed in $\mathcal{O}(n \cdot m^2)$ time, where $m$ is the number of vertices of $K$.
\label{lem:LemmaVolumeEnclosingBox}
\end{lem}
\begin{proof}[Proof by induction over $n$]
In two dimensions, the minimal enclosing rectangle has at most twice the area of the contained polygon.

In higher dimensions $n$, let $P,Q$ be two points of $K$ with maximum distance and $\vert PQ\vert = l$. Let $\pi_p$ be the hyperplane normal to $\overline{PQ}$ in the point $P$. Let $K'$ be the orthogonal projection of $K$ onto $\pi_p$. By the inductive hypothesis there is a ($n-1$)-dimensional box $B'$ containing $K'$ for which
\begin{equation*}
V'(B') \leq (n-1)!V'(K')
\end{equation*}
where $V'$ denotes the ($n-1$)-dimensional volume. Then $K$ is contained in the box $B$ with base $B'$ and height $l$.
\begin{align*}
V(B) = l V'(B') \leq l(n-1)!V'(K') 
\end{align*}
It is well known (see e.g.~\cite{convexRegions}) that for any convex body $K$, its projection $K'$ on some hyperplane $\pi_P$, and a line segment $l$ perpendicular to $\pi_P$, it holds: $V\left(K\right) \geq \frac{1}{n}\cdot l \cdot V'\left(K'\right)$. Hence, we get for the volume of $B$:
\begin{equation*}
V(B) \leq n! V(K)
\end{equation*}
$B$ can be computed by testing every pair of vertices to find $P$ and $Q$ that have maximal distance. This takes $\mathcal{O}(m^2)$ time. Then $K$ gets projected on a plane perpendicular to $\overline{PQ}$. This is possible in $\mathcal{O}(m)$ time. Then we proceed recursively with the projection of $K$. In total we need $\mathcal{O}(n\cdot m^2)$ time.
\end{proof}
The construction in the proof of Lemma~\ref{lem:LemmaVolumeEnclosingBox} is the same as in Lemma~7 from~\cite{convexRegions}. We get a total running time of $\mathcal{O}\left(m^2\right)$ for computing the bounding boxes of three-dimensional polyhedra with $m$ vertices in total since the dimension is fixed.

For the analysis of the algorithm presented in this section we need several notations and lemmata that follow. Consider the box $b =(h,w,d)$ with $h \geq w \geq d$ obtained from the polyhedron $p$ by Lemma~\ref{lem:LemmaVolumeEnclosingBox}. Notice that in every facet of $b$ lies at least one point contained in $p$. We call the top and bottom one $T$ and $B$. In the left and right facet of $b$, we choose a point each and call them $L$ and $R$. By construction, the distance from them to the front facet has to be the same. We do the same for the front and rear facet and call them $F$ and $D$ respectively. We know from the construction that $\vert TB \vert = h$ and $\overline{TB}$ is parallel to the longest edge of $b$. If we project the polyhedron onto a plane perpendicular to $\overline{TB}$, we call the images of $T$, $L$, $R$, $F$ and $D$ under the projection $T'$, $L'$, $R'$, $F'$ and $D'$, respectively. See Figure~\ref{fig:BoxWithProjection} for illustration. Due to the construction of $b$, $\vert L'R' \vert = w$ holds.

\begin{figure}{ht}
\centering
\includegraphics[height=0.35\textwidth]{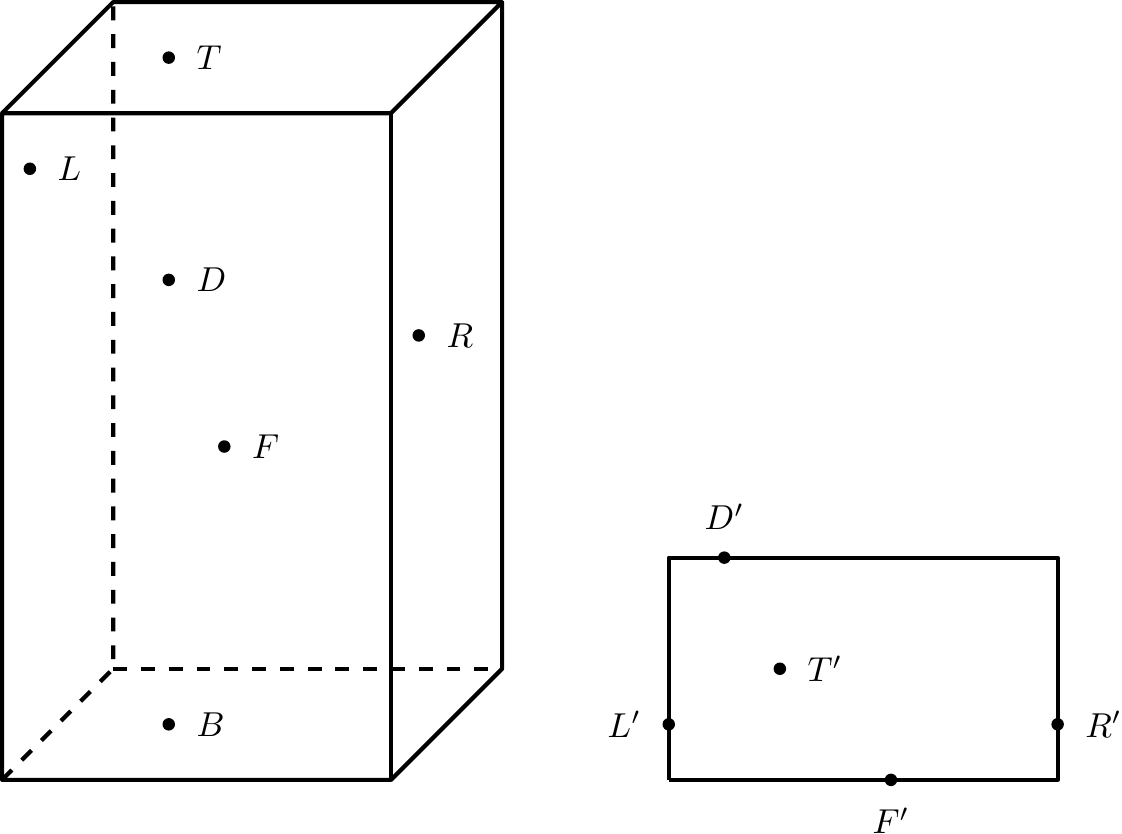}
\caption{Box with a point of the enclosed polyhedron in every facet and the projection of the box on a plane perpendicular to $\overline{TB}$. By construction, the images of $T$ and $B$ under the projection are the same.}

\label{fig:BoxWithProjection}
\end{figure}
\begin{lem}
Let $b =(h,w,d)$ with $h \geq w \geq d$ be the enclosing box obtained for polyhedron $p$. Then, parallel to any given plane, $p$ contains a line segment of length at least $w\cdot \frac{1}{\sqrt{5}}$.
\label{lem:TriangleHeightw}
\end{lem}
\begin{proof}
Consider the points $T$, $B$, $L$ and $R$ as described above. The distance between line segment $\overline{TB}$ and $L$ or the distance between line segment $\overline{TB}$ and $R$ is at least $\frac{w}{2}$. Let w.l.o.g. $L$ be the point with largest distance to $\overline{TB}$. Consider the triangle $\bigtriangleup(T,B,L)$ with labeled edges and angles according to Figure~\ref{fig:TBL}. Notice that $\alpha \leq 90\degree$ and $\beta \leq 90\degree$. Let $a_t$ be the height of the triangle on edge $t$, $a_b$ on edge $b$ and $a_l$ on edge $l$. 

\begin{figure}[ht]
    \centering
    \begin{subfigure}[b]{0.47\textwidth}
        \includegraphics[width=\textwidth]{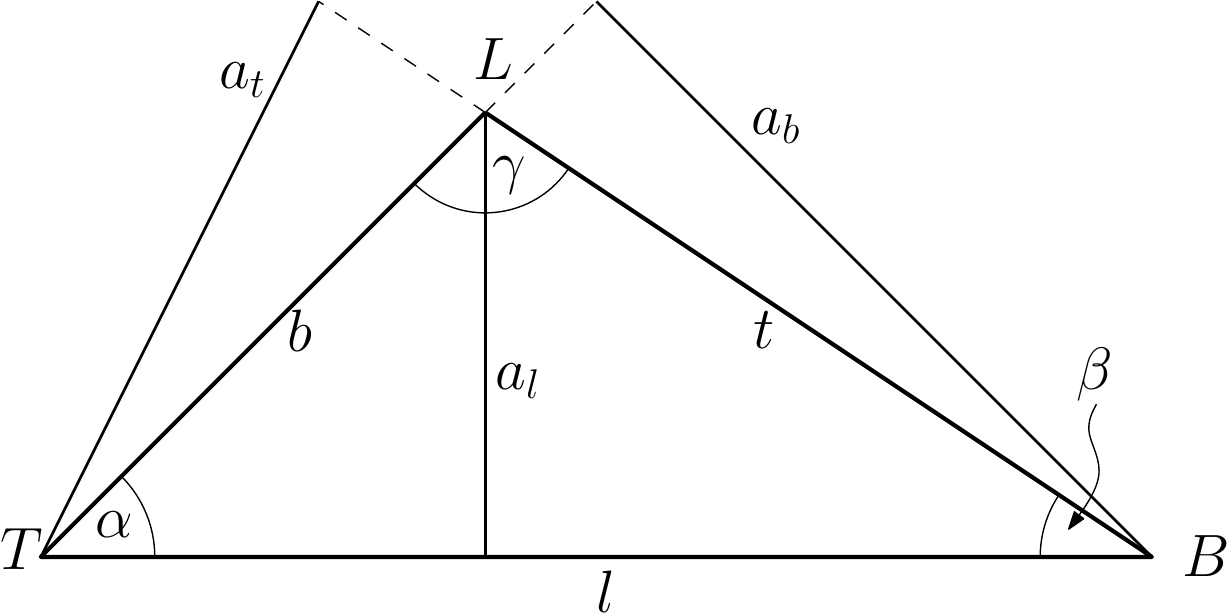}
        \caption{Labelled triangle $\bigtriangleup(T,B,L)$}
        \label{fig:TBL}
    \end{subfigure}
      \qquad
    \begin{subfigure}[b]{0.43\textwidth}
        \includegraphics[width=\textwidth]{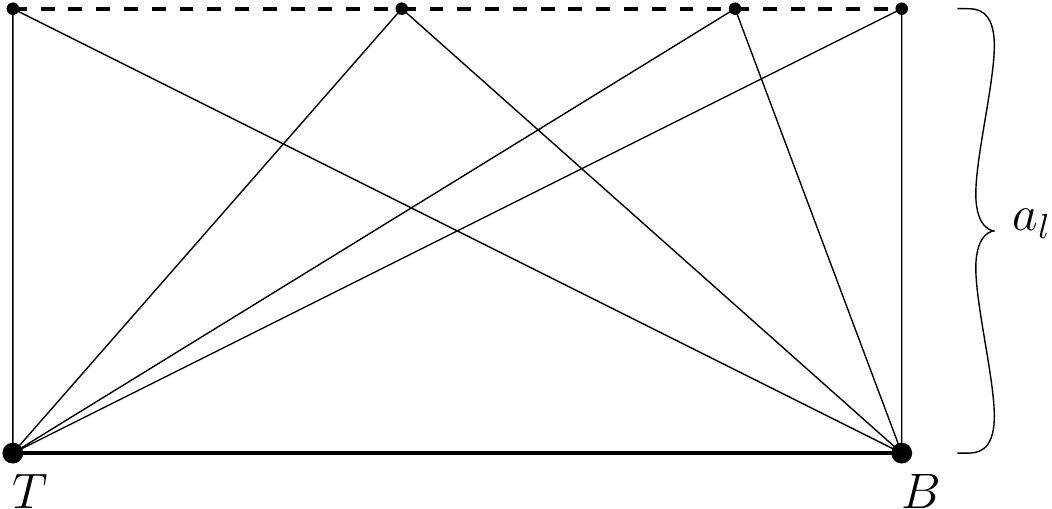}
        \caption{Possible triangles $\bigtriangleup(T,B,L)$}
        \label{fig:possibleTBL}
    \end{subfigure}
    
    \caption{}
    
\end{figure}

Due to the construction of $\bigtriangleup(T,B,L)$, we know that $a_l \geq \frac{w}{2}$. We will later show that $a_b \geq \frac{w}{\sqrt{5}}$ and $a_t \geq \frac{w}{\sqrt{5}}$. If we choose a plane parallel to the given one, such that the intersection between the plane and $\bigtriangleup(T,B,L)$ contains $T$, $B$ or $L$ but is not only one point, then we know that the intersection is at least a line segment with length $\min\left(a_t, a_b,a_l\right) \geq \frac{w}{\sqrt{5}}$ which completes the proof. It remains to show that $a_t, a_b \geq \frac{w}{\sqrt{5}}$.

We only show that $ a_b \geq \frac{w}{\sqrt{5}}$ since the proof for $a_t$ is analogous. Figure~\ref{fig:possibleTBL} depicts possible triangles with given distance $\vert TB\vert$ and height $a_l$. $a_b$ is the distance between $B$ and the line defined by $T$ and $L$. Since $\beta \leq 90\degree$ this distance is minimal for
$\beta = 90\degree$. 

Let $A$ be the area of $\bigtriangleup(T,B,L)$ with $\beta = 90\degree$.
\begin{align*}
&&\frac{a_l \cdot \vert TB \vert}{2} = &A = \frac{a_b \cdot \vert TL\vert}{2}\\
&\text{hence }& a_l \cdot h &= a_b \cdot \sqrt{h^2 + a_l^2} &&\text{, since } \vert TB\vert = h\\ &&&&&\text{ and using Pythagoras's theorem for }\vert TL \vert\\
&\text{i.e. }& a_b &= \frac{a_l \cdot h}{\sqrt{h^2 + a_l^2}}\\
&&&= \frac{1}{\sqrt{\frac{1}{a_l^2}+ \frac{1}{h^2}}}\\
&&&\geq \frac{1}{\sqrt{\frac{4}{w^2}+ \frac{1}{w^2}}}\\
&&&= \frac{w}{\sqrt{5}}
\end{align*}
\end{proof}

\begin{lem}
Let $b =(h,w,d)$ with $h \geq w \geq d$ be the enclosing box obtained for a polyhedron $p$. Then the projection of $p$ onto an arbitrary line $g$ has length at least $\frac{1}{8\sqrt{3}}d$.
\label{lem:projectionLengthd}
\end{lem}
\begin{proof}
We construct four line segments inside of $p$ such that the projection of at least one of them onto the line has the desired length.

Consider the projection of $p$ onto a plane perpendicular to $\overline{TB}$ as described above (Figure~\ref{fig:BoxWithProjection}). Then $\bigtriangleup (L',R',F')$ or $\bigtriangleup (L',R',D')$ has an area $A \geq \frac{dw}{4}$. The perimeter of the projection of the box, namely $2 (w+d)$, gives an upper bound for the perimeter $U$ of the triangles. It is well known (see ,e.g.,~\cite{Coxeter}) that the radius of a triangle with area $A$ and perimeter $U$ is $r = \frac{2A}{U}$. Hence, we know that the projection of $p$ contains a circle with radius $r$ where
\begin{equation*}
r = \frac{2A}{U} \geq \frac{dw}{4(d+w)} \geq \frac{1}{8} d \text{ , since }d\leq w.
\end{equation*}
See Figure~\ref{fig:CircleInTriangle} for an example.

\begin{figure}
    \centering
    \begin{subfigure}[b]{0.57\textwidth}
        \includegraphics[width=\textwidth]{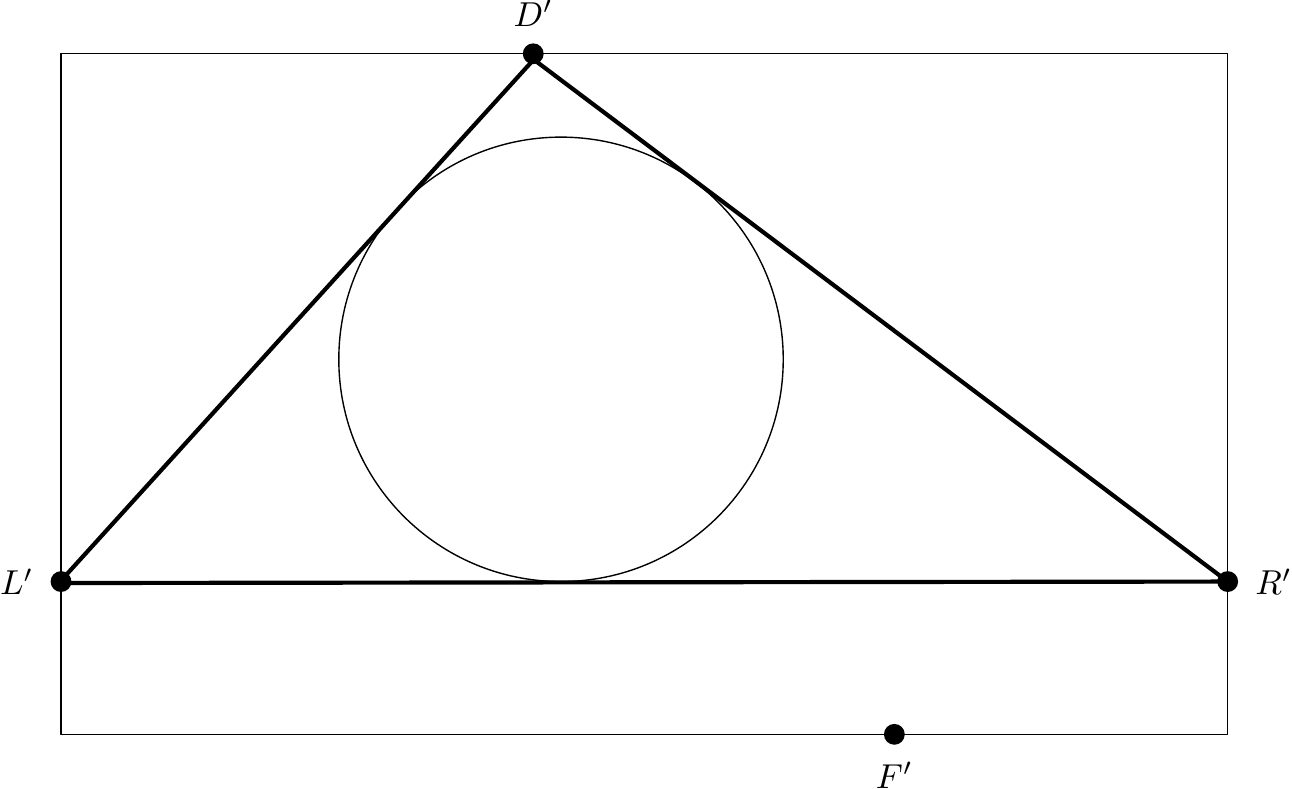}
        \caption{Circle in the projection of $p$ that has radius at least $\frac{1}{8} d$}
        \label{fig:CircleInTriangle}
    \end{subfigure}
      \qquad
    \begin{subfigure}[b]{0.3\textwidth}
        \includegraphics[width=\textwidth]{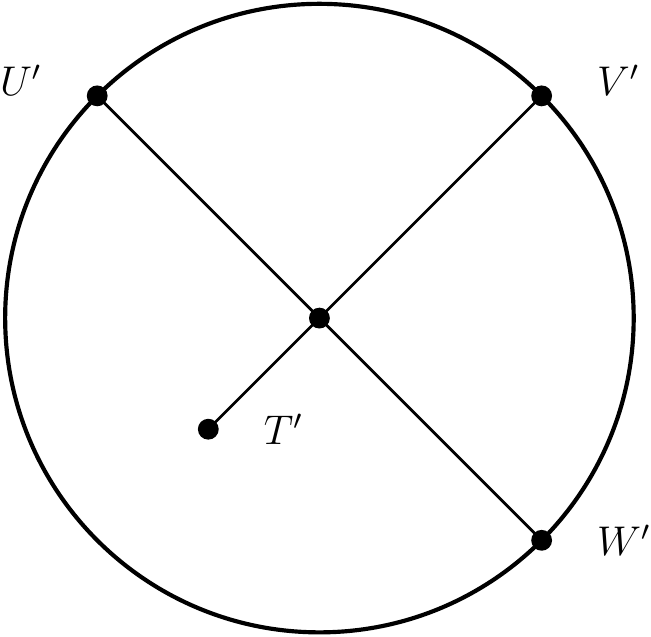}
        \caption{Construction of $U'$, $V'$ and $W'$}
        \label{fig:UVW}
    \end{subfigure}
    \caption{}
    
\end{figure}

Now we can find points $U'$, $V'$, $W'$ in the projection, such that $U'$, $V'$, $W'$ lie on the circle with radius $r$ and $\vert T'V'\vert = k \geq r$, $\vert U'W' \vert = l = 2r$ and $\overline{T'V'} \perp \overline{U'W'}$. To obtain $V'$, we shoot a ray from $T'$ through the center of the circle until we hit the circle and call this point $V'$. $\overline{U'W'}$ is the diameter of the circle perpendicular to $\overline{T'V'}$. See Figure~\ref{fig:UVW} for an example.

Let $U$, $V$, $W$ be preimages of $U'$, $V'$, $W'$ under the projection. Hence, they lie inside $p$. The line segments whose projections on the given line $g$ we consider are $\overline{BT}$, $\overline{BV}$, $\overline{VT}$ and $\overline{WU}$. 

The length of the projection of a line segment onto $g$ is the scalar product of the vector between the endpoints of the line segment and a unit vector with same direction as $g$. To simplify the computation of the scalar product, we define the coordinate system as follows: $B$ is equal to the origin. $T$ lies on the z-axis. The y-coordinate of $V$ is $0$. Then $U$ and $W$ have the same x-coordinate. Now we have
\begin{align*}
\overrightarrow{BT} &= \left( \begin{array}{c} 0 \\ 0 \\ h \end{array} \right)
&\overrightarrow{BV} &= \left( \begin{array}{c} k \\ 0 \\ h_V \end{array} \right)
&\overrightarrow{VT} &= \left( \begin{array}{c} -k \\ 0 \\ h-h_V \end{array} \right)
&\overrightarrow{WU} &= \left( \begin{array}{c} 0 \\ l \\ h_{WU} \end{array} \right),
\end{align*}
for some values $k$, $l$, $h_V$, $h_{WU}$ where $0\leq h_V \leq h$ and $\vert h_{WU} \vert \leq h$. Let $\overrightarrow{g} =\left( \begin{array}{c} x \\ y \\ z \end{array} \right)$ be the direction of $g$ in the defined coordinate system, with $\vert\overrightarrow{g}\vert = 1$. We now look at the lengths of the projections of the line segments onto the given line.

\begin{component}[Case 1: $\vert x \vert \geq \frac{1}{\sqrt{3}}$]
Then
\begin{align*}
\vert \overrightarrow{BV} \cdot \overrightarrow{g}\vert \geq k \vert x \vert \geq \frac{1}{\sqrt{3}\cdot 8} d
\end{align*}
or
\begin{equation*}
\vert \overrightarrow{VT} \cdot \overrightarrow{g}\vert \geq k \vert x \vert \geq \frac{1}{\sqrt{3}\cdot 8} d
\end{equation*}
\end{component}
\begin{component}[Case 2: $\vert z \vert\cdot h \geq \frac{1}{\sqrt{3}\cdot 8}d$]
\begin{equation*}
\vert \overrightarrow{BT} \cdot \overrightarrow{g}\vert = h \cdot \vert z \vert \geq \frac{1}{\sqrt{3}\cdot 8}d
\end{equation*}
\end{component}
\begin{component}[Case 3: $\vert y \vert \geq \frac{1}{\sqrt{3}}$ and $\sgn(y) = \sgn(h_{WU}z)$]
\begin{align*}
\vert \overrightarrow{WU} \cdot \overrightarrow{g}\vert \geq l \vert y \vert \geq \frac{1}{\sqrt{3}\cdot 8} d
\end{align*}
\end{component}
\begin{component}[Case 4: $\vert y \vert \geq \frac{1}{\sqrt{3}}$ and $\sgn(y) \neq \sgn(h_{WU}z)$ and $\vert z \vert\cdot h < \frac{1}{\sqrt{3}\cdot 8}d$]
$ $\newline Note: $\vert h_{WU}z \vert \leq h \vert z \vert < \frac{1}{\sqrt{3}\cdot 8}d$ and $ l \vert y \vert \geq \frac{2}{\sqrt{3}\cdot 8}d$
\begin{equation*}
\vert\overrightarrow{WU} \cdot \overrightarrow{g}\vert = l \vert y \vert - \vert h_{WU} z \vert \geq \frac{1}{\sqrt{3}\cdot 8} d
\end{equation*}
\end{component}
Since $\vert\overrightarrow{g}\vert = 1$, $\vert x \vert \geq \frac{1}{\sqrt{3}}$ or $\vert y \vert \geq \frac{1}{\sqrt{3}}$ or $\vert z \vert \geq \frac{1}{\sqrt{3}}$ holds. Hence, at least one of the 4 cases occurs because $h \geq d$ .
\end{proof}

Consider the polyhedra $p_1, p_2, p_3$ that determine $h_{\max}$, $w_{\max}$ and $d_{\max}$ in the placement the described algorithm computes. $p_1$ contains a line segment of length $h_{\max}$ and so its projections of at the least one of the axes is at least $\frac{1}{\sqrt{3}}h_{\max}$. W.l.o.g. let this axis be the x-axis. Then by Lemma~\ref{lem:TriangleHeightw} the projection of $p_2$ onto the y-z-plane contains a line of length at least $\frac{1}{\sqrt{5}}w_{\max}$. Therefore, the projection of $p_2$ onto the y-axis or the one onto the z-axis has length at least $\frac{1}{\sqrt{2}}\cdot \frac{1}{\sqrt{5}}w_{\max} = \frac{1}{\sqrt{10}}w_{\max}$. The projection of $p_3$ on the remaining axis has length at least $\frac{1}{8\sqrt{3}}d_{\max}$ by Lemma~\ref{lem:projectionLengthd}. An axis parallel box with minimal volume containing $p_1, p_2, p_3$ has at least the described side lengths and so we get the following lemma:
\begin{lem}
For packing convex polyhedra under rigid motions into a minimum-volume axis parallel container, the following inequality holds:
\begin{equation*}
h_{\max}\cdot w_{\max}\cdot d_{\max} \leq 24\sqrt{10} V_{\text{opt}}.
\end{equation*}
\label{lem:hmax_wmax_dmax_convex_polytopes}
\end{lem}
From Lemma~\ref{lem:LemmaVolumeEnclosingBox} we know that the volume of the smallest enclosing box for a polyhedron is at most 6 times the volume of the polyhedron. With the previous lemma and this knowledge we derive the following approximation ratio from inequality~(\ref{eq:getAssympAppr}):
\begin{equation}
 \frac{12c}{\left(1-\varepsilon\right)\left(c-1\right)} + \frac{c \cdot 24\sqrt{10}}{\varepsilon}.
 \label{eq:approxPackPolyhedra}
\end{equation}
The running time of this algorithm is determined by the computation of the bounding boxes and the packing of these boxes: $\mathcal{O}\left(m^2 + n\log n\right)$ where $m$ is the total number of vertices of the polyhedra. Hence, we get by minimizing term~(\ref{eq:approxPackPolyhedra}) as before the following theorem.
\begin{thm}
The given algorithm computes an orthogonal container with volume at most $277.59$ times the volume of an orthogonal minimal container for the variant of three-dimensional \textsc{OMCOP} where $n$ convex polyhedra having $m$ vertices in total are to be packed under rigid motions in time $\mathcal{O}\left(m^2 + n\log n\right)$.
\end{thm}

\subsubsection*{Convex Container}
If we allow  arbitrary convex containers instead of axis parallel boxes we get the following lemma instead of Lemma~\ref{lem:hmax_wmax_dmax_convex_polytopes}:
\begin{lem}
For packing convex polyhedra under rigid motions into a minimum-volume convex container, the following inequality holds:
\begin{equation*}
h_{\max}\cdot w_{\max}\cdot d_{\max} \leq 24\sqrt{60} V_{\text{opt}}.
\end{equation*}
\end{lem}
\begin{proof}
As before let $p_1,p_2,p_3$ be the polytopes that determine $h_{\max}$, $w_{\max}$ and $d_{\max}$. $p_1$ contains a line segment of length $h_{\max}$. By Lemma~\ref{lem:TriangleHeightw} $p_2$ contains a line segment of length $\frac{w_{\max}}{\sqrt{5}}$ that is perpendicular to the first line segment. By Lemma~\ref{lem:projectionLengthd} $p_3$ contains a line segment with length $\frac{d_{\max}}{8\sqrt{3}}$ that is perpendicular to the first two lines. Since any convex body containing three pairwise perpendicular line segments of length $a, b, c$ has volume at least $\frac{1}{6} abc$ (cf. Lemma 6 in \cite{convexRegions}), we get a lower bound on the volume of the convex hull which is also a lower bound for the volume of an optimal container.
\end{proof}

As before we use Lemma~\ref{lem:LemmaVolumeEnclosingBox} and the previous lemma to estimate inequality~(\ref{eq:getAssympAppr}) and optimize the following approximation ratio:
\begin{equation*}
 \frac{12c}{\left(1-\varepsilon\right)\left(c-1\right)} + \frac{c \cdot 24\sqrt{60}}{\varepsilon}.
\end{equation*}
By minimizing this term as before, we yield the following result.
\begin{thm}
The algorithm given in Section~\ref{sec:ConvPolyOMCOP} computes a convex container with volume at most $511.37$ times the volume of a minimal convex container for packing $n$ convex polyhedra having $m$ vertices in total under rigid motions in time $\mathcal{O}\left(m^2 + n\log n\right)$.
\end{thm}

\bibliographystyle{plain}
	\bibliography{references}
		
\end{document}